%% file: paper.tex
\documentclass[runningheads]{llncs}
\usepackage[utf8]{inputenc}
\usepackage[numbers]{natbib}
\usepackage{amsmath}
\usepackage{amsfonts}
\usepackage{amssymb}
\usepackage{graphicx}
\usepackage{siunitx}
\usepackage{thm-restate}
\usepackage[font=small,labelfont=bf]{caption}
\usepackage{algorithm}
\usepackage[noend]{algpseudocode}
\usepackage{macros}
\usepackage{booktabs}
\usepackage{multirow}

\title{Datalog Materialisation in Distributed RDF Stores with Dynamic Data Exchange}
\titlerunning{Distributed Datalog Materialisation}

\author{
    Temitope Ajileye \and
    Boris Motik \and
    Ian Horrocks
}
\authorrunning{T. Ajileye et al.}
\institute{Department of Computer Science, University of Oxford, Oxford, UK}

\begin{document}

\maketitle

\begin{abstract}
\input{abstract}
\end{abstract}

\input{introduction}
\input{preliminaries}
\input{motivation}
\input{algorithm}
\input{evaluation}
\input{conclusion}

\section*{Acknowledgments}

This work was supported by the SIRIUS Centre for Scalable Access in the Oil and
Gas Domain, and the EPSRC project AnaLOG.

\bibliographystyle{splncsnat}
\bibliography{references,ours}

\newpage
\appendix
\input{appendix}

\end{document}

%% file: abstract.tex
Several centralised RDF systems support datalog reasoning by precomputing and
storing all logically implied triples using the well-known \emph{semina\"{i}ve
algorithm}. Large RDF datasets often exceed the capacity of centralised RDF
systems, and a common solution is to distribute the datasets in a cluster of
shared-nothing servers. While numerous distributed query answering techniques
are known, distributed semina\"{i}ve evaluation of arbitrary datalog rules is
less understood. In fact, most distributed RDF stores either support no
reasoning or can handle only limited datalog fragments. In this paper we extend
the \emph{dynamic data exchange} approach for distributed query answering by
\citet{pmnh18dynamic-rdf-stores} to a reasoning algorithm that can handle
arbitrary rules while preserving important properties such as nonrepetition of
inferences. We also show empirically that our algorithm scales well to very
large RDF datasets.

%% file: introduction.tex
\section{Introduction}\label{sec:introduction}

Reasoning with datalog rules over RDF data plays a key role on the Semantic
Web. Datalog can capture the structure of an application domain using if-then
rules, and OWL 2 RL ontologies can be translated into datalog rules. Datalog
reasoning is supported in several RDF management systems such as Oracle's
database \cite{DBLP:conf/semweb/KolovskiWE10},
GraphDB,\footnote{\url{http://graphdb.ontotext.com/}} Amazon
Neptune,\footnote{http://aws.amazon.com/neptune/} VLog
\cite{DBLP:conf/aaai/UrbaniJK16}, and
RDFox.\footnote{\url{http://www.cs.ox.ac.uk/isg/tools/RDFox/}} All of these
system use a \emph{materialisation} approach to reasoning, where all facts
implied by the dataset and the rules are precomputed and stored in a
preprocessing step. This is usually done using the \emph{semina{\"i}ve
algorithm} \cite{abiteboul95foundation}, which ensures the \emph{nonrepetition
property}: no rule is applied to the same facts more than once.

Many RDF management systems are \emph{centralised} in that they store and
process all data on a single server. To scale to workloads that cannot fit into
a single server, it is common to distribute the data in a cluster of
interconnected, shared-nothing servers and use a distributed query answering
strategy. \citet{DBLP:journals/pvldb/AbdelazizHKK17} present a comprehensive
survey of 22 approaches to distributed query answering, and
\citet{pmnh18dynamic-rdf-stores} discuss several additional systems. There is
considerable variation between these approaches: some use data replication,
some compute joins on dedicated server, others use distributed join algorithms,
and many leverage big data frameworks such as Hadoop and Spark for data storage
and query processing. In contrast, distributed datalog materialisation is less
well understood, and it is more technically challenging. Freshly derived facts
must be stored so that they can be taken into account in future rule
applications, but without repeating derivations. Moreover, synchronisation
between rule applications should be reduced to allow parallel computation.

Several theoretical frameworks developed in the 90s aim to address these
questions \cite{DBLP:journals/jlp/GangulyST92, DBLP:journals/tkde/WolfsonO93,
DBLP:conf/pods/SeibL91, DBLP:conf/pdis/ShaoBH91, DBLP:journals/tkde/ZhangWC95}.
As we discuss in more detail in Section~\ref{sec:motivation}, they constrain
the rules so that that each server performs only certain rule applications, and
they send the derived facts to all servers where these facts could participate
in further rule applications. Thus, same facts can be stored on more than one
server, which severely limits the scalability of such systems.

The Semantic Web community has recently developed several RDF-specific
approaches. a number of them are hardwired to fixed datalog rules, such as RDFS
\cite{DBLP:conf/semweb/WeaverH09, DBLP:conf/semweb/KaoudiMK08} or so-called
\emph{ter Horst fragment} \cite{DBLP:journals/ws/UrbaniKMHB12,
DBLP:conf/ipps/GuWWYH15}. Focusing on a fixed set of rules considerably
simplifies the technical problems. PLogSPARK \cite{DBLP:conf/wise/WuLWYWZ16}
and SPOWL \cite{DBLP:conf/sigmod/LiuM17} handle arbitrary rules, but they do
not seem to use semina{\"i}ve evaluation. Finally, several probabilistic
algorithms aim to handle large datasets \cite{DBLP:journals/ws/OrenKASTH09,
DBLP:conf/sigmod/LiuM17}, but these approaches are approximate and are thus
unsuitable for many applications. Distributed SociaLite
\cite{DBLP:journals/pvldb/SeoPSL13} is the only system we are aware of that
provides semina{\"i}ve evaluation for arbitrary datalog rules. It uses a custom
graph model, but the approach can readily be adapted to RDF. Moreover, its
rules must explicitly encode the communication and storage strategy, which
increases complexity.

In this paper we present a new technique for distributed materialisation of
arbitrary datalog rules. Unlike SociaLite, we do not require any distributed
processing hints in the rules. We also do not duplicate any data and thus
remove an obstacle to scalability. Our approach is based on the earlier work by
\citet{pmnh18dynamic-rdf-stores} on distributed query answering using
\emph{dynamic data exchange}, from which it inherits several important
properties. First, inferences that can be made within a single server do not
require any communication; coupled with careful data partitioning, this can
very effectively minimise network communication. Second, rule evaluation is
completely asynchronous, which promotes parallelism. This, however, introduces
a complication: to ensure nonrepetition of inferences, we must be able to
partially order rule derivations across the cluster, which we achieve using
\emph{Lamport timestamps} \cite{DBLP:journals/cacm/Lamport78}. We discuss the
motivation and the novelty in more detail in Section~\ref{sec:motivation}, and
in Section~\ref{sec:algorithm} we present the approach formally.

We have implemented our approach in a new prototype system called DMAT. In
Section~\ref{sec:evaluation} we present the results of an empirical evaluation.
We compared DMAT with WebPIE \cite{DBLP:journals/ws/UrbaniKMHB12}, investigated
how it scales with increasing data loads, and compared it with RDFox to
understand the impact of distribution on concurrency. Our results show that
DMAT outperforms WebPIE by an order of magnitude (albeit with some differences
in the setting), and that it can handle well increasing data loads; moreover,
DMAT's performance is comparable to that of RDFox on a single server. Our
algorithms are thus a welcome addition to the techniques for implementing truly
scalable semantic systems.

%% file: preliminaries.tex
\section{Preliminaries}\label{sec:preliminaries}

We now recapitulate the syntax and the semantics of RDF and datalog. A
\emph{constant} (aka \emph{RDF term}) is an IRI, a blank node, or a
literal. Datalog constants can be arbitrary sets, but we are limiting them to RDF terms because this work in within the context of RDF stores. A \emph{term} is a constant or a \emph{variable}. An \emph{atom} $a$
(aka \emph{triple pattern}) has the form ${a = \triple{t_s}{t_p}{t_o}}$ over
terms $t_s$ (\emph{subject}), $t_p$ (\emph{predicate}), and $t_o$
(\emph{object}). A \emph{fact} (aka \emph{triple}) is an variable-free atom. A
\emph{dataset} (aka \emph{RDF Graph}) is a finite set of facts.

We define the set of \emph{positions} as ${\Pi = \{ s, p, o \}}$. Then, for ${a
= \triple{t_s}{t_p}{t_o}}$ and ${\pi \in \Pi}$, we define ${\term{a}{\pi} =
t_\pi}$---that is, $\term{a}{\pi}$ is the term that occurs in $a$ at position
$\pi$. A \emph{substitution} $\sigma$ is a partial function that maps finitely
many variables to constants. For $\alpha$ a term or an atom, $\alpha\sigma$ is
the result of replacing with $\sigma(x)$ each occurrence of a variable $x$ in
$\alpha$ on which $\sigma$ is defined.

A \emph{query} $Q$ is a conjunction of atoms ${a_1 \wedge \dots \wedge a_n}$.
Substitution $\sigma$ is an \textit{answer} to $Q$ on a dataset $I$ if
${a_i\sigma \in I}$ holds for each ${1 \leq i \leq n}$.

A datalog \textit{rule} $r$ is an implication of the form ${h \leftarrow b_1
\wedge \dots \wedge b_n}$, where $h$ is the \emph{head} atom, all $b_i$ are
\emph{body} atoms, and each variable occurring in $h$ also occurs in some
$b_i$. A datalog \textit{program} is a finite set of rules. Let $I$ be a
dataset. The result of applying $r$ to $I$ is ${r(I) = I \cup \{ h\sigma \mid
\sigma \text{ is an answer to } b_1 \wedge \dots \wedge b_n \text{ on } I \}}$.
For $P$ a program, let ${P(I) = \bigcup_{r \in P} r(I)}$; let ${P^0(I) = I}$;
and let ${P^{i+1}(I) = P(P^i(I))}$ for ${i \geq 0}$. Then, ${P^\infty(I) =
\bigcup_{i \geq 0} P^i(I)}$ is the \emph{materialisation} of $P$ on $I$. This
paper deals with the problem of computing $P^\infty(I)$ where $I$ is
distributed across of a cluster of servers such that each fact is stored in
precisely one server.

%% file: motivation.tex
\section{Motivation and Related Work}\label{sec:motivation}

We can compute $P^\infty(I)$ using the definition in
Section~\ref{sec:preliminaries}: we evaluate the body of each rule ${r \in P}$
as a query over $I$ and instantiate the head of $r$ for each query answer, we
eliminate duplicate facts, and we repeat the process until no new facts can be
derived. However, since ${P^i(I) \subseteq P^{i+1}(I)}$ holds for each ${i \geq
0}$, such a \emph{na{\"i}ve} approach repeats in each round of rule
applications the work from all previous rounds. The \emph{sem\"{i}naive
strategy} \cite{abiteboul95foundation} avoids this problem: when matching a
rule $r$ in round $i+1$, at least one body atom of $r$ must be matched to a
fact derived in round $i$. We next discuss now these ideas are implemented in
the existing approaches, and then we present an overview of our approach.

\subsection{Related Approaches}

Several approaches to distributed reasoning partition rule applications across
servers. For example, to evaluate rule ${\triple{x}{R}{z} \leftarrow
\triple{x}{R}{y} \wedge \triple{y}{R}{z}}$ on $\ell$ servers, one can let each
server $i$ with ${1 \leq i \leq \ell}$ evaluate rule
\begin{align}
    \triple{x}{R}{z} \leftarrow \triple{x}{R}{y} \wedge \triple{y}{R}{z} \wedge h(y) = i,   \label{eq:part-rule}
\end{align}
where $h(y)$ is a \emph{partition function} that maps values of $y$ to integers
between $1$ and $\ell$. If $h$ is uniform, then each server receives roughly
the same fraction of the workload, which benefits parallelisation. However,
since a triple of the form $\triple{s}{R}{o}$ can match either atom in the body
of \eqref{eq:part-rule}, each such triple must be replicated to servers $h(s)$
and $h(o)$ so they can participate in rule applications. Based on this idea,
\citet{DBLP:journals/jlp/GangulyST92} show how to handle general datalog;
\citet{DBLP:journals/tkde/ZhangWC95} study different partition functions;
\citet{DBLP:conf/pods/SeibL91} identify programs and partition functions where
no replication of derived facts is needed; \citet{DBLP:conf/pdis/ShaoBH91}
further break rules in segments; and \citet{DBLP:journals/tkde/WolfsonO93}
replicate all facts to all servers. The primary motivation behind these
approaches seems to be parallelisation of computation, which explains why the
high rates of data replication were not seen as a problem. However, high
replication rates are not acceptable when data distribution is used to increase
a system's capacity.

Materialisation can also be implemented without any data replication. First,
one must select a triple partitioning strategy: a common approach is to assign
each $\triple{s}{p}{o}$ to server $h(s)$ for a suitable hash function $h$, and
another popular option is to use a distributed file system (e.g., HDFS) and
thus leverage its partitioning mechanism. Then, one can evaluate the rules
using a suitable distributed query algorithm and distribute the newly derived
triples using the partitioning strategy. These principles were used to realise
RDFS reasoning \cite{DBLP:conf/semweb/WeaverH09, DBLP:conf/semweb/KaoudiMK08},
and they are also implicitly present in approaches implemented on top of big
data frameworks such as Hadoop \cite{DBLP:journals/ws/UrbaniKMHB12} and Spark
\cite{DBLP:conf/ipps/GuWWYH15, DBLP:conf/wise/WuLWYWZ16,
DBLP:conf/sigmod/LiuM17}. However, most of these can handle only fixed rule
sets, which considerably simplifies algorithm design. For example,
semina{\"i}ve evaluation is not needed in the RDFS fragment since the
nonrepetition of inferences can be ensured by evaluating rules in a particular
order \cite{DBLP:conf/ipps/GuWWYH15}. PLogSPARK \cite{DBLP:conf/wise/WuLWYWZ16}
and SPOWL \cite{DBLP:conf/sigmod/LiuM17} handle arbitrary rules using the
na{\"i}ve algorithm, which can be detrimental when programs are moderately
complex.

Distributed SociaLite \cite{DBLP:journals/pvldb/SeoPSL13} is the only system
known to us that implements distributed semina{\"i}ve evaluation for general
datalog. It requires users to explicitly specify the data distribution strategy
and communication patterns. For example, by writing a fact $R(a,b)$ as
$R[a](b)$, one specifies that the fact is to be stored on server $h(a)$ for
some function $h$. Rule \eqref{eq:part-rule} can then be written in SociaLite
as ${R[x](z) \leftarrow R[x](y) \wedge R[y](z)}$, specifying that the rule
should be evaluated by sending each fact $R[a](b)$ to server $h(b)$, joining such
facts with $R[b](c)$, and sending the resulting facts $R[a](c)$ to server
$h(a)$. While the evaluation of some rules is parallelised, servers must
synchronise after each round of rule application.

\subsection{Dynamic Data Exchange for Query Answering}\label{sec:motivation:dynamic}

Before describing our approach to distributed datalog materialisation, we next
recapitulate the earlier work by \citet{pmnh18dynamic-rdf-stores} on
distributed query answering using \emph{dynamic data exchange}, which provides
the foundation for this paper.

This approach to query answering assumes that all triples are partitioned into
$\ell$ mutually disjoint datasets ${I_1, \dots, I_\ell}$, with $\ell$ being the
number of servers. The main objectives of dynamic exchange are to reduce
communication and eliminate synchronisation between servers. To achieve the
former goal, each server $k$ maintains three \emph{occurrence mappings}
$\occur{k}{s}$, $\occur{k}{p}$, and $\occur{k}{o}$. For each resource $r$
occurring in $I_k$, set $\occur{k}{s}(r)$ contains all servers where $r$ occurs
in the subject position, and $\occur{k}{p}(r)$ and $\occur{k}{o}(r)$ provide
analogous information for the predicate and object positions. To understand how
occurrences are used, consider evaluating ${Q = \triple{x}{R}{y} \wedge
\triple{y}{R}{z}}$ over datasets ${I_1 = \{ \triple{a}{R}{b}, \triple{b}{R}{c}
\}}$ and ${I_2 = \{ \triple{b}{R}{d}, \triple{d}{R}{e} \}}$. Both servers
evaluate $Q$ using index nested loop joins. Thus, server 1 evaluates
$\triple{x}{R}{y}$ over $I_1$, which produces a \emph{partial} answer
${\sigma_1 = \{ x \mapsto a, y \mapsto b \}}$. Server 1 then evaluates
${\triple{y}{R}{z}\sigma_1 = \triple{b}{R}{z}}$ over $I_1$ and thus obtains one
full answer ${\sigma_2 = \{ x \mapsto a, y \mapsto b, z \mapsto c \}}$. To see
whether $\triple{b}{R}{z}$ can be matched on other servers, server 1 consults
its occurrence mappings for all resources in the atom. Since ${\occur{1}{s}(b)
= \occur{1}{p}(R) = \{ 1,2 \}}$, server 1 sends the partial answer $\sigma_1$
to server 2, telling it to continue matching the query. After receiving
$\sigma_1$, server 2 matches atom $\triple{b}{R}{z}$ in $I_2$ to obtain another
full answer ${\sigma_3 = \{ x \mapsto a, y \mapsto b, z \mapsto d \}}$.
However, server 2 also evaluates $\triple{x}{R}{y}$ over $I_2$, obtaining
partial answer ${\sigma_4 = \{ x \mapsto b, y \mapsto d \}}$, and it consults
its occurrences to determine which servers can match ${\triple{y}{R}{z}\sigma_4
= \triple{d}{R}{z}}$. Since ${\occur{2}{s}(d) = \{ 2 \}}$, server 2 knows it is
the only one that can match this atom, so it proceeds without any communication
and computes ${\sigma_5 = \{ x \mapsto b, y \mapsto d, z \mapsto e \}}$.

This strategy has several important benefits. First, all answers that can be
produced within a single server, such as $\sigma_5$ in our example, are
produced without any communication. Second, the location of every resource is
explicitly recorded, rather than computed using a fixed rule (e.g., a hash
function). We use this to partition a graph based on its structural properties
and thus collocate highly interconnected resources. Combined with the first
property, this can significantly reduce network communication. Third, the
system is completely asynchronous: when server 1 sends $\sigma_1$ to server 2,
it does not need to to wait for server 2 to finish; and server 2 can process
$\sigma_1$ whenever it wishes. This eliminates the need for synchronisation
between servers, which is beneficial for parallelisation.

\subsection{Our Contribution}\label{sec:motivation:contribution}

In this paper we extend the dynamic data exchange framework to datalog
materialisation. We draw inspiration from the work by
\citet{mnpho14parallel-materialisation-RDFox} on parallelising datalog
materialisation in centralised, shared memory systems. Intuitively, their
algorithm considers each triple in the dataset, identifies each rule and body
atom that can be matched to the triple, and evaluates the rest of the rule as a
query. This approach is amenable to parallelisation since distinct processors
can simultaneously process distinct triples; since the number of triples is
generally very large, the likelihood of workload skew between processors is
very low.

Our distributed materialisation algorithm is based on the same general
principle: each server matches the rules to locally stored triples, but the
resulting queries are evaluated using dynamic data exchange. This approach
requires no synchronisation between servers, and it reduces communication in
the same way as described in Section~\ref{sec:motivation:dynamic}. We thus
expect our approach to exhibit the same good properties as the approach to
query answering by \citet{pmnh18dynamic-rdf-stores}.

The lack of synchronisation between servers introduces a technical
complication. Remember that, to avoid repeating derivations, at least one body
atom in a rule must be matched to a fact derived in the previous round of rule
application. However, due to asynchronous rule application, there is no global
notion of a rule application round (unlike, say, SociaLite). A na{\"i}ve
solution would be to associate each fact with a timestamp recording when the
fact has derived, so the order of fact derivation could then be recovered by
comparing timestamps. However, this would require maintaining a high coherence
of server clocks in the cluster, which is unrealistic in practice. Instead, we
use Lamport timestsamps \cite{DBLP:journals/cacm/Lamport78}, which provide a
cheap way of determining a partial order of events across a cluster. We
describe this technique in more detail in Section~\ref{sec:algorithm}.

Another complication is due to the fact that the occurrence mappings stored in
the servers may need to be updated due to the derivation of new triples. For
completeness, it is critical that all servers are updated before such triples
are used in rule applications. Our solution to this problem is fully
asynchronous, which again benefits parallelisation.

Finally, since no central coordinator keeps track of the state of the
computation of different servers, detecting when the system as a whole can
terminate is not straightforward. We solve this problem using a well-known
termination detection algorithm based on token passing
\cite{DBLP:journals/ipl/DijkstraFG83}.

%% file: algorithm.tex
\section{Distributed Materialisation Algorithm}\label{sec:algorithm}

We now present our distributed materialisation algorithm and prove it to be
correct. We present the algorithm in steps. In
Section~\ref{sec:algorithm:timestamps} we discuss data structures that the
servers use to store their triples and implement Lamport timestamps. In
Section~\ref{sec:algorithm:occurrence} we discuss the occurrence mappings. In
Section~\ref{sec:algorithm:communication} we discuss the communication
infrastructure and the message types used. In
Section~\ref{sec:algorithm:algorithm} we present the algorithm's pseudocode. In
Section~\ref{sec:algorithm:termination} we discuss how to detect termination.
Finally, in Section~\ref{sec:algorithm:correctness} we argue about the
algorithm's correctness.

\subsection{Adding Lamport Timestamps to Triples}\label{sec:algorithm:timestamps}

As already mentioned, to avoid repeating derivations, our algorithm uses
Lamport timestamps \cite{DBLP:journals/cacm/Lamport78}, which is a technique
for establishing a causal order of events in a distributed system. If all
servers in the system could share a global clock, we could trivially associate
each event with a global timestamp, which would allow us to recover the
`happens-before' relationship between events by comparing timestamps. However,
maintaining a precise global clock in a distributed system is technically very
challenging, and Lamport timestamps provide a much simpler solution. In
particular, each event is annotated an integer timestamp in a way that
guarantees the following property ($\ast$):
\begin{quote}
    if there is any way for an event $A$ to possibly influence an event $B$,
    then the timestamp of $A$ is strictly smaller then the timestamp of $B$.
\end{quote}
To achieve this, each server maintains a local integer clock that is
incremented each time an event of interest occurs, which clearly ensures
($\ast$) if $A$ and $B$ occur within one server. Now assume that $A$ occurs in
server $s_1$ and $B$ occurs in $s_2$; clearly, $A$ can influence $B$ only if
$s_1$ sends a message to $s_2$, and $s_2$ processes this message before event
$B$ takes place. To ensure that holds ($\ast$) in such a case, server $s_1$
includes its current clock value into the message it sends to $s_2$; moreover,
when processing this message, server $s_2$ updates its local clock to the
maximum of the message clock and the local clock, and then increments the local
clock. Thus, when $B$ happens after receiving the message, it is guaranteed to
have a timestamp that is larger than the timestamp of $A$.

To map this idea to datalog materialisation, a derivation of a fact corresponds
to the notion of an event, and using a fact to derive another fact corresponds
to the `influences' notion. Thus, we associates facts with integer timestamps.

More precisely, each server $k$ in the cluster maintains an integer $C_k$
called the \emph{local clock}, a set $I_k$ of the derived triples, and a
partial function ${T_k : I_k \rightarrow \mathbb{N}}$ that associates triples
with natural numbers. Function $T_k$ is partial because timestamps are not
assigned to facts upon derivation, but during timestamp synchronisation. Before
the algorithm is started, $C_k$ must be initialised to zero, and all input
facts (i.e., the facts given by the user) partitioned to server $k$ should be
loaded into $I_k$ and assigned a timestamp of zero.

To capture formally how timestamps are used during query evaluation, we
introduce the notion of an \emph{annotated query} as a conjunction of the form
\begin{align}
    Q = a_1^{\bowtie_1} \wedge \dots \wedge a_n^{\bowtie_n},    \label{eq:annotated-query}
\end{align}
where each $a_i$ is an atom and each $\bowtie_i$ is symbol $<$ or symbol
$\leq$; each $a_i^{\bowtie_i}$ is called an \emph{annotated atom}. An annotated
query requires a timestamp to be evaluated. More precisely, a substitution
$\sigma$ is an answer to $Q$ on $I_k$ and $T_k$ w.r.t.\ a timestamp $\tau$ if
(i)~$\sigma$ is an answer to the `ordinary' query ${a_1 \wedge \dots \wedge
a_n}$ on $I_k$, and (ii)~for each ${1 \leq i \leq n}$, the value of $T_k$ is
defined for $a_i\sigma$ and it satisfies ${T_k(a_i\sigma) \bowtie \tau}$. For
example, let $Q$, $I$, and $T$ be as follows, and let ${\tau = 2}$:
\begin{displaymath}
\begin{array}{@{}r@{\;}l@{}}
    Q   & = \triple{x}{R}{y}^< \wedge \triple{y}{S}{z}^\leq \qquad I = \{ \triple{a}{R}{b}, \triple{b}{S}{c}, \triple{b}{S}{d}, \triple{b}{S}{e} \} \\
    T   & = \{ \triple{a}{R}{b} \mapsto 1, \; \triple{b}{S}{c} \mapsto 2, \; \triple{b}{S}{d} \mapsto 3 \} \\
\end{array}
\end{displaymath}
Then, ${\sigma_1 = \{ x \mapsto a, y \mapsto b, z \mapsto c \}}$ is an answer
to $Q$ on $I$ and $T$ w.r.t.\ $\tau$, whereas ${\sigma_2 = \{ x \mapsto a, y
\mapsto b, z \mapsto d \}}$ and ${\sigma_3 = \{ x \mapsto a, y \mapsto b, z
\mapsto e \}}$ are not. $\sigma_2$ is excluded because $ T(\triple{b}{S}{d}) \geq 2$ and $\sigma_3$ is excluded because  the timestamp of $\triple{b}{S}{e}$ is undefined.

To incorporate this notion into our algorithm, we assume that each server can
evaluate a single annotated atom. Specifically, given an annotated
$a^{\bowtie}$, a timestamp $\tau$, and a substitution $\sigma$, server $k$ can
call ${\textsc{Evaluate}(a^{\bowtie}, \tau, I_k, T_k, \sigma)}$. The call
returns each substitution $\rho$ defined over the variables in $a$ and $\sigma$
such that ${\sigma \subseteq \rho}$ holds, ${a\rho \in I_k}$ holds, and $T_k$
is defied on $a\rho$ and it satisfies ${T(a\rho) \bowtie \tau}$. In other
words, $\textsc{Evaluate}$ matches $a^{\bowtie}$ in $I_k$ and $T_k$ w.r.t.\
$\tau$ and it returns each extension of $\sigma$ that agrees with
$a^{\bowtie}$. For efficiency, server $k$ should index the facts in $I_k$; any
RDF indexing scheme can be used, and one can modify index lookup to simply skip
over facts whose timestamps do not match $\tau$.

Finally, we describe how rule matching is mapped to answering annotated
queries. Let $P$ be a datalog program to be materialised. Given a fact $f$,
function ${\textsc{MatchRules}(f, P)}$ considers each rule ${h \leftarrow b_1
\wedge \dots \wedge b_n \in P}$ and each body atom $b_p$ with ${1 \leq p \leq
n}$, and, for each substitution $\sigma$ over the variables of $b_p$ where ${f
= b_p\sigma}$, it returns ${(\sigma, b_p, Q, h)}$ where $Q$ is the annotated
query
\begin{align}
    b_1^< \wedge \dots \wedge b_{p-1}^< \wedge b_{p+1}^\leq \wedge \dots \wedge b_n^\leq. \label{eq:annotated-query-for-p}
\end{align}

Intuitively, $\textsc{MatchRules}$ identifies each rule and each \emph{pivot}
body atom $b_p$ that can be matched to $f$ via substitution $\sigma$. This
$\sigma$ will be extended to all body atoms of the rule by matching all
remaining atoms in nested loops using function $\textsc{Evaluate}$. The
annotations in \eqref{eq:annotated-query-for-p} specify how to match the
remaining atoms without repetition: facts matched to atoms before (resp.\
after) the pivot must have timestamps strictly smaller (resp.\ smaller or
equal) than the timestamp of $f$. As is usual in query evaluation, the atoms of
\eqref{eq:annotated-query-for-p} may need to be reordered to obtain an
efficient query plan. This can be achieved using any known technique, and
further discussion of this issue is out of scope of this paper.

\subsection{Occurrence Mappings}\label{sec:algorithm:occurrence}

To decide whether rule matching may need to proceed on other servers, each
server $k$ must store indexes $\mu_{k,s}$, $\mu_{k,p}$, and $\mu_{k,o}$ called
\emph{occurrence mappings} that map resources to sets of server IDs. To ensure
scalability, $\mu_{k,s}$, $\mu_{k,p}$, and $\mu_{k,o}$ need not be defined on
all resources: if, say, $\mu_{k,s}$ is not defined on resource $r$, we will
assume that $r$ can occur on any server. However, these mappings will need to
be correct during algorithm's execution: if server $I_k$ contains a resource
$r$ (in any position), and if $r$ occurs on some other server $j$ in position
$\pi$, then $\mu_{k,\pi}$ must be defined on $r$ and it must contain $j$.
Moreover, all servers will have to know the locations of all resources
occurring in the heads of the rules in $P$.

Storing only partial occurrences at each server introduces a complication: when
a server processes a partial match $\sigma$ received from another server, its
local occurrence mappings may not cover some of the resources in $\sigma$.
\citet{pmnh18dynamic-rdf-stores} solve this by accompanying each partial match
$\sigma$ with a vector ${\pmb \lambda = \lambda_s, \lambda_p, \lambda_o}$ of
\emph{partial occurrences}. Whenever a server extends $\sigma$ by matching an
atom, it also records in $\pmb \lambda$ its local occurrences for each resource
added to $\sigma$ so that this information can be propagated to subsequent
servers.

Occurrence mappings are initialised on each server $k$ for each resource that
initially occurs in $I_k$, but they may need to be updated as fresh triples are
derived. To ensure that the occurrences correctly reflect the distribution of
resources at all times, occurrence mappings of all servers must be updated
\emph{before} a triple can be added to the set of derived triples of the target
server.

Our algorithm must decide where to store each freshly derived triple. It is
common practice in distributed RDF systems to store all triples with the same
subject on the same server. This is beneficial since it allows subject--subject
joins---the most common type of join in practice---to be answered without any
communication. We follow this well-established practice and ensure that the
derived triples are grouped by subject. Consequently, we require that
$\mu_{k,s}(r)$, whenever it is defined, contains exactly one server. Thus, to
decide where to store a derived triple, the server from the subject's
occurrences is used, and, if the subject occurrences are unavailable, then a
predetermined server is used.

\subsection{Communication Infrastructure and Message Types}\label{sec:algorithm:communication}

We assume that the servers can communicate asynchronously by passing messages.
That is, each server can call ${\textsc{Send}(m, d)}$ to send a message $m$ to
a destination server $d$. This function can return immediately, and the
receiver can processes the message later. Also, our core algorithm is correct
as long as each sent message is processed eventually, regardless of whether the
messages are processed in the order in which they are sent between servers. We
next describe the two types of message used in our algorithm. The approach used
to detect termination can introduce other message types and might place
constraints on the order of message delivery; we discuss this in more detail in
Section~\ref{sec:algorithm:termination}.

Message $\PAR{i, \sigma, Q, h, \tau, \pmb \lambda}$ informs a server that
$\sigma$ is a partial match obtained by matching some fact with timestamp
$\tau$ to the body of a rule with head atom $h$; moreover, the remaining atoms
to be matched are given by an annotated query $Q$ starting from the atom with
index $i$. The partial occurrences for all resources mentioned in $\sigma$ are
recorded in $\pmb \lambda$.

Message $\FCT{f, D, k_h, \tau, \pmb \lambda}$ says that $f$ is a freshly
derived fact that should be stored at server $k_h$. Set $D$ contains servers
whose occurrences must be updated due to the addition of $f$. Timestamp $\tau$
corresponds to the time at which the message was sent. Finally, $\pmb \lambda$
are the partial occurrences for the resources in $f$.

\citet{pmnh18dynamic-rdf-stores} already observed $\mathsf{PAR}$ messages
correspond to partial join results so a large number of such messages can be
produced during query evaluation; moreover, to facilitate asynchronous
processing, the received $\mathsf{PAR}$ messages may need to be buffered on the
received server, which can easily require excessive space. They also presented a flow
control mechanism that can be used to restrict memory consumption at each server
without jeopardising completeness. This solution is directly applicable to our
problem as well, so we do not discuss it any further.

\subsection{The Algorithm}\label{sec:algorithm:algorithm}

With these definitions in mind, Algorithms~\ref{alg:materialisation-i}
and~\ref{alg:materialisation-ii} comprise our approach to distributed datalog
materialisation. Before starting, each server $k$ loads its subset of the input
RDF graph into $I_k$, sets the timestamp of each fact in $I_k$ to zero,
initialises $C_k$ to zero, and receives the copy of the program $P$ to be
materialised. The server then starts an arbitrary number of server threads, each
executing the $\textsc{ServerThread}$ function. Each thread repeatedly
processes either an unprocessed fact $f$ in $I_k$ or an unprocessed message
$m$; if both are available, they can be processed in arbitrary order.
Otherwise, the termination condition is processed as we discuss later in
Section~\ref{sec:algorithm:termination}.

Function $\textsc{Synchronise}$ updates the local clock $C_k$ with a timestamp
$\tau$. This must be done in a critical section (i.e., two threads should not
execute it simultaneously). The local clock is updated if ${C_k \leq \tau}$
holds; moreover, all facts in $I_k$ without a timestamp are are timestamped
with $C_k$ since they are derived before the event corresponding to $\tau$.
Assigning timestamps to facts in this way reduces the need for synchronising
access to $C_k$ between threads.

Function $\textsc{ProcessFact}$ kickstarts the matching of the rules to fact
$f$. After synchronising the clock with the timestamp of $f$, it simply calls
$\textsc{MatchRules}$ to identify all rules where one atom matches to $f$, and
forwards each match $\textsc{FinishMatch}$ to finish matching the pivot atom.

A $\mathsf{PAR}$ message is processed by matching atom $a_i^{\bowtie_i}$ of the
annotated query in $I_k$ and $T_k$ w.r.t.\ $\tau$, and forwarding each match to
$\textsc{FinishMatch}$.

A $\mathsf{FCT}$ message informs server $k$ that fact $f$ will be added to
the set $I_{k_h}$ of facts derived at server $k_h$. Set $D$ lists all remaining
servers that need to be informed of the addition, and partial occurrences $\pmb
\lambda$ are guaranteed to correctly reflect the occurrences of each resource
in $f$. Server $k$ updates its $\mu_{k,\pi}(r)$ by appending $\lambda_\pi(r)$
(line~\ref{process-message-fact-add-occurrence}). Since servers can
simultaneously process $\mathsf{FCT}$ messages, server $k$ adds to $D$ all
servers that might have been added to $\mu_{k,\pi}(r)$ since the point when
$\lambda_\pi(r)$ had been constructed
(line~\ref{process-message-fact-extend-D}), and it also updates
$\lambda_\pi(r)$ (line~\ref{process-message-fact-add-occurrence}). Finally, the
server adds $f$ to $I_k$ if $k$ is the last server
(line~\ref{process-message-add}), and otherwise it forwards the message to
another server $d$ form $D$.

Function $\textsc{FinishMatch}$ finishes matching atom $a_\mathit{last}$ by
(i)~extending $\pmb \lambda$ with the occurrences of all resources that might
be relevant for the remaining body atoms or the rule head, and (ii)~either
matching the next body atom or deriving the rule head. For the former task, the
algorithm identifies in line~\ref{finish-match-add-supporting-ocr} each
variable $x$ in the matched atom that either occurs in the rule head or in a
remaining atom, and for each $\pi$ it adds the occurrences of $x\sigma$ to
$\lambda_\pi$. Now if $Q$ has been matched completely
(line~\ref{finish-match-matched-start}), the server also ensures that the
partial occurrences are correctly defined for the resources occurring in the
rule head (lines~\ref{finish-match-h-start}--\ref{finish-match-h-end}), it
identifies the server $k_h$ that should receive the derived fact as described
in Section~\ref{sec:algorithm:occurrence}, it identifies the set $D$ of the
destination servers whose occurrences need to be updated, and it sends the
$\mathsf{FCT}$ message to one server from $D$. Otherwise, atom $a_{i+i}\sigma$
must be matched next. To determine the set $D$ of servers that could possibly
match this atom, server $k$ intersects the occurrences of each resource from
$a_{i+i}\sigma$ (line~\ref{finish-match-destination-list}) and sends a
$\mathsf{PAR}$ message to all servers in $D$.

\begin{algorithm}[tb] \caption{Distributed Materialisation Algorithm at Server $k$}\label{alg:materialisation-i} \begin{algorithmiccont}

\Function{ServerThread}{}
	\While{cannot terminate}
		\If{$I_k$ contains an unprocessed fact $f$, or a message $m$ is pending}
			\State \Call{ProcessFact}{$f$} or \Call{ProcessMessage}{$m$}, as appropriate
		\ElsIf{the termination token has been received}
			\State Process the termination token
		\EndIf
	\EndWhile
\EndFunction

\vspace{2ex}

\Function{ProcessFact}{$f$}                                                                                                                             \label{process-fact}
	\State \Call{Synchronise}{$T_k(f)$}                                                                                                                 \label{process-fact-sync}
	\For{\textbf{each} $(\sigma, a, Q, h) \in \textsc{MatchRules}(f,P)$}                                                                                \label{process-fact-match-rules}
		\State \Call{FinishMatch}{$0, \sigma, a, Q, h, T_k(f), \pmb \emptyset$}                                                                         \label{process-fact-finish-match}
	\EndFor
\EndFunction

\vspace{2ex}

\Function{ProcessMessage}{$\PAR{i, \sigma, Q, h, \tau, \pmb\lambda}$} where $Q = a_1^{\bowtie_1} \wedge \dots \wedge a_n^{\bowtie_n}$                   \label{process-message-partial}
	\State \Call{Synchronise}{$\tau$}                                                                                                                   \label{process-message-partial-sync}
	\For{\textbf{each} substitution $\sigma' \in \textsc{Evaluate}(a_i^{\bowtie_i}, \tau, I_k, T_k, \sigma)$}                                           \label{process-message-eval}
		\State \Call{FinishMatch}{$i, \sigma', a_i, Q, h, \tau, \pmb \lambda$}                                                                          \label{process-message-finish-match}
	\EndFor
\EndFunction

\vspace{2ex}

\Function{ProcessMessage}{$\FCT{f, D, k_h, \tau, \pmb\lambda}$}                                                                                         \label{process-message-fact}
	\State \Call{Synchronise}{$\tau$}                                                                                                                   \label{process-message-fact-sync}
	\For{\textbf{each} resource $r$ in $f$ and each position $\pi \in \Pi$}
		\State $D \defeq D \cup \big[\mu_{k,\pi}(r) \setminus \lambda_{\pi}(r)\big]$                                                                    \label{process-message-fact-extend-D}
		\State $\lambda_{\pi}(r) \defeq \mu_{k,\pi}(r) \defeq \lambda_{\pi}(r) \cup \mu_{k,\pi}(r)$                                                     \label{process-message-fact-add-occurrence}
	\EndFor
	\If{$D = \emptyset$}
		Add $f$ to $I_k$                                                                                                                                \label{process-message-add}
	\Else
 		\State Remove an element $d$ from $D$, preferring any element over $k_h$ if possible
 		\State \Call{Send}{$\FCT{f, D, k_h, C_k, \pmb\lambda},d$}                                                                                       \label{process-message-fact-gettime}
 	\EndIf
\EndFunction

\vspace{2ex}

\Function{Synchronise}{$\tau$} (must be executed in a critical section)                                                                                 \label{sync-server}
	\If{$C_k \leq \tau$}
		\For{\textbf{each} fact $f \in I_k$ such that $T_k$ is undefined on $f$}
            $T_k(f) \defeq C_k$                                                                                                                         \label{sync-server-add-fact}
		\EndFor
        \State $C_k \defeq \tau + 1$
	\EndIf
\EndFunction

\end{algorithmiccont}
\end{algorithm}

\begin{algorithm}[tb]
\caption{Distributed Materialisation Algorithm at Server $k$ (Continued)}\label{alg:materialisation-ii}
\begin{algorithmiccont}

\Function{FinishMatch}{$i, \sigma, a_\mathit{last}, Q, h, \tau, \pmb\lambda$} where $Q = a_1^{\bowtie_1} \wedge \dots \wedge a_n^{\bowtie_n}$           \label{finish-match}
    \For{\textbf{each} var.\ $x$ occurring in $a_\mathit{last}$ and in $h$ or $a_j$ with $j>i$, and each $\pi \in \Pi$}
        \State Extend $\lambda_{\pi}$ with the mapping $x\sigma \mapsto \mu_{k,\pi}(x\sigma)$                                                           \label{finish-match-add-supporting-ocr}
    \EndFor
	\If{$i = n$}                                                                                                                                        \label{finish-match-matched-start}
        \For{\textbf{each} resource $r$ occurring in $h$ and each $\pi \in \Pi$}                                                                        \label{finish-match-h-start}
            \State Extend $\lambda_{\pi}$ with the mapping $r \mapsto \mu_{k,\pi}(r)$
        \EndFor                                                                                                                                         \label{finish-match-h-end}
		\State $k_h \defeq$ the owner server for the derived fact
		\State $D \defeq \{ k_h \}$                                                                                                                     \label{finish-match-D-n}
		\For{\textbf{each} position $\pi \in \Pi$ and $r = \term{h\sigma}{\pi}$ where $k_h \not\in \lambda_{\pi}(r)$}
			\State Add $k_h$ to $\lambda_{\pi}(r)$
			\For{\textbf{each} $\pi' \in \Pi$}
				Add $\lambda_{\pi'}(r)$ to $D$
			\EndFor
		\EndFor
		\State Remove an element $d$ from $D$, preferring any element over $k_h$ if possible                                                            \label{finish-match-send-fact}
		\If{$d = k$}
			\Call{ProcessMessage}{$\FCT{h\sigma, D, k_h, C_k, \pmb\lambda}$}                                                                            \label{finish-match-process-FCT}
		\Else{}
			\Call{Send}{$\FCT{h\sigma, D, k_h, C_k, \pmb\lambda}, d$}                                                                                   \label{finish-match-send-FCT}
		\EndIf                                                                                                                                          \label{finish-match-matched-end}
	\Else
		\State $D \defeq$ the set of all servers                                                                                                        \label{finish-match-D-start}
		\For{\textbf{each} position $\pi \in \Pi$ where $\term{a_{i+1}\sigma}{\pi}$ is a resource $r$}
		    $D \defeq D \cap \lambda_{\pi}(r)$                                                                                                          \label{finish-match-destination-list}
		\EndFor                                                                                                                                         \label{finish-match-D-end}
		\For{\textbf{each} $d \in D$}
			\If{$d=k$}
				\Call{ProcessMessage}{$\PAR{i+1, \sigma, Q, h, \tau, \pmb\lambda}$}
			\Else{}
				\Call{Send}{$\PAR{i+1, \sigma, Q, h, \tau, \pmb\lambda}, d$}
			\EndIf
		\EndFor
	\EndIf
\EndFunction

\end{algorithmiccont}
\end{algorithm}

\subsection{Termination Detection}\label{sec:algorithm:termination}

Since no server has complete information about the progress of any other
server, detecting termination is nontrivial; however, we can reuse an existing
solution.

When messages between each pair of servers are guaranteed to be delivered in
order in which they are sent (as is the case in our implementation), one can
use Dijkstra's token ring algorithm \cite{DBLP:journals/ipl/DijkstraFG83},
which we summarise next. All servers in the cluster are numbered from $1$ to
$\ell$ and are arranged in a ring (i.e., server 1 comes after server $\ell$).
Each server can be black or white, and the servers will pass between them a
\emph{token} that can also be black or white. Initially, all servers are white
and server 1 has a white token. The algorithm proceeds as follows.
\begin{itemize}
    \item When server 1 has the token and it becomes idle (i.e., it has no
    pending work or messages), it sends a white token to the next server in the
    ring.
    
    \item When a server other than 1 has the token and it becomes idle, the
    server changes the token's colour to black if the server is itself black
    (and it leaves the token's colour unchanged otherwise); the server forwards
    the token to the next server in the ring; and the server changes its colour
    to white.

    \item A server $i$ turns black whenever it sends a message to a server ${j
    < i}$.
    
    \item All servers can terminate when server 1 receives a white token.
\end{itemize}

The Dijkstra--Scholten algorithm extends this approach to the case when the
order of message delivery cannot be guaranteed.

\subsection{Correctness}\label{sec:algorithm:correctness}

We next prove that our algorithm is correct and that it exhibits the
nonrepetition property. We present here only an outline of the correctness
argument, and we give the full proof in the supplementary material.

Let us fix a run of Algorithms~\ref{alg:materialisation-i}
and~\ref{alg:materialisation-ii} on some input. First, we show that Lamport
timestamps capture the causality of fact derivation in this run. To this end,
we introduce four event types relating to an arbitrary fact $f$. Event
$\mathsf{add}_k(f)$ occurs when $f$ is assigned a timestamp on server $k$ in
line~\ref{sync-server-add-fact}. Event $\mathsf{process}_k(f)$ occurs when
server $k$ starts processing a new fact in line~\ref{process-fact-sync}. Event
$\mathsf{PAR}_k(f,i)$ occurs when server $k$ completes
line~\ref{process-message-partial-sync} for a $\mathsf{PAR}$ message with index
$i$ originating from a call to $\textsc{MatchRules}$ on fact $f$. Finally,
event $\mathsf{FCT}_k(f)$ occurs when server $k$ completes
line~\ref{process-message-fact-sync} for a $\mathsf{FCT}$ message for fact $f$.
We write ${e_1 \rightsquigarrow e_2}$ if event $e_1$ occurs chronologically
before event $e_2$; this relation is clearly transitive and irreflexive.
Furthermore, each fact is stored and assigned a timestamp on just one server,
so we define $T(f)$ as $T_k(f)$ for the unique server $k$ that satisfies ${f
\in I_k}$. Lemma~\ref{lemma:ccc} then essentially says that the happens-before
relationship between facts and events on facts agrees with the timestamps assigned to the
facts.

\begin{restatable}{lemma}{lemmaccc}\label{lemma:ccc}
    In each run of the algorithm, for each server $k$, and all facts $f_1$ and
    $f_2$, we have ${T(f_1) < T(f_2)}$ whenever one of the following holds:
    \begin{itemize}
        \item ${\mathsf{PAR}_k(f_1,i) \rightsquigarrow \mathsf{add}_k(f_2)}$
        for some $i$,

        \item ${\mathsf{process}_k(f_1) \rightsquigarrow \mathsf{FCT}_k(f_2)}$,
        or
        
        \item ${\mathsf{PAR}_k(f_1,i) \rightsquigarrow \mathsf{FCT}_k(f_2)}$
        for some $i$.
    \end{itemize}
\end{restatable}

Next, we show that then the occurrence mappings $\mu_{k,\pi}$ on each relevant
server $k$ are updated whenever a triple is added to some $I_j$. This condition
is formally captured in Lemma~\ref{lemma:occ}, and it ensures that partial
answers are sent to all relevant servers that can possibly match an atom in a
query. Note that the implication in Lemma~\ref{lemma:occ} is the only relevant
direction: if $\mu_{k,\pi}(r)$ contains irrelevant servers, we can have
redundant $\mathsf{PAR}$ messages, but this does not harm correctness.

\begin{restatable}{lemma}{lemmaocc}\label{lemma:occ}
    At any point in the algorithm's run, for all servers $k$ and $j$, each
    position ${\pi \in \Pi}$, and each resource $r$ such that $r$ occurs in
    $I_j$ at position $\pi$ and $\mu_{k,\pi}$ is defined on $r$, property ${j
    \in \mu_{k,\pi}(r)}$ holds.
\end{restatable}

Using Lemmas~\ref{lemma:ccc} and~\ref{lemma:occ}, we prove our main claim.

\begin{restatable}{theorem}{thmcorrectness}\label{thm:correctness}
    For ${I_1, \dots, I_\ell}$ the sets obtained by applying
    Algorithms~\ref{alg:materialisation-i} and~\ref{alg:materialisation-ii} to
    an input set of facts $I$ and program $P$, we have ${P^{\infty}(I) = I_1
    \cup \dots \cup I_\ell}$. Moreover, the algorithm exhibits the
    nonrepetition property.
\end{restatable}

%% file: evaluation.tex
\section{Evaluation}\label{sec:evaluation}

To evaluate the practical applicability of our approach, we have implemented a
prototype distributed datalog reasoned that we call DMAT. We have reused a
well-known centralised RDF system to store and index triples in RAM, on top of
which we have implemented a mechanism for associating triples with timestamps.
To implement the $\textsc{Evaluate}$ function, we use the system's interface
for answering individual atoms and then simply filter out the answers whose
timestamp does not match the given one. For simplicity, DMAT currently uses
only one thread per server, but we plan to remove this limitation in future. We have published\footnote{http://krr-nas.cs.ox.ac.uk/2019/distributed-materialisation/} the executable and test files used, with the exception of the datasets, which can be recreated using the LUBM generator.

We have evaluated our system's performance in three different ways, each aimed
at analysing a specific aspect of the problem. First, to establish a baseline
for the performance of DMAT, as well as to see whether distributing the data
can speed up the computation, we compared DMAT with RDFox
\cite{mnpho14parallel-materialisation-RDFox}---a state-of-the-art, centralised,
RAM-based reasoner---on a relatively small dataset. Second, to compare the
performance of our approach with the state-of-the-art for distributed
reasoning, we compared DMAT with WebPIE
\cite{DBLP:journals/ws/UrbaniKMHB12}---a distributed RDF reasoner based on
MapReduce. Third, we studied the scalability of our approach by proportionally
increasing the input data and the number of servers.

Few truly large RDF datasets are publicly available, so the evaluation of
distributed reasoning is commonly based on the well-known
LUBM\footnote{\url{http://swat.cse.lehigh.edu/projects/lubm/}} benchmark (e.g.,
\cite{DBLP:journals/ws/UrbaniKMHB12, DBLP:conf/sigmod/LiuM17,
DBLP:conf/wise/WuLWYWZ16}). Following this practice, LUBM datasets of sizes
ranging from 134~M to 6.5~G triples provided the input data for our system. We
also used the \emph{lower bound} program was obtained by extracting the OWL 2
RL portion of the LUBM ontology and translating it into datalog.

We conducted all tests with DMAT on the Amazon Elastic Compute Cloud (EC2). We
used the $r4.8xlarge$ servers, each equipped with a 2.3~GHz Intel Broadwell
processors and 244~GB RAM; such a large amount of RAM was needed since the
underlying storage mechanism in our system is RAM-based. An additional,
identical server stored the dictionary (i.e., a data structure mapping
resources to integers): this server did not participate in materialisation, but
was used only to distribute the program and the data to the cluster. Finally, the EC2 instances offer 10 Gbps network performance, according to the manifest published by Amazon\footnote{https://aws.amazon.com/ec2/instance-types/}. In all
tests apart from the ones with WebPIE, we partitioned the dataset by using the graph
partitioning approach by \citet{pmnh18dynamic-rdf-stores}: this data
partitioning approach aims to place strongly connected resources on the same
server and thus reduce communication overhead. For the tests with WebPIE, due to memory constraints of the partitioning software, we partitioned triples by
subject hashing. For each test, we loaded the input triples and the program
into all servers, and computed the materialisation while recording the
wall-clock time. Apart from reporting this time, we also report the
\emph{reasoning throughput} measured in thousands of triples derived per second
and worker (ktps/w). We next discuss the results of our experiments.

\paragraph{Comparison with RDFox.} First, we ran RDFox and DMAT on a fixed
dataset while increasing the number of threads for RDFox and the numbers of
servers for DMAT. Since RDFox requires the materialised dataset to fit into RAM
of a single server, we used a small input dataset of just 134~M triples. The
results, shown in Table~\ref{tab:concurrency-v-distribution}, provide us with
two insights. First, the comparison on one thread establishes a baseline for
the DMAT's performance. In particular, DMAT is slower than RDFox, which is not
surprising: RDFox is a mature and tuned system, whereas DMAT is just a
prototype. However, DMAT is still competitive with RDFox, suggesting that our
approach is free of any overheads that might make it uncompetitive. Second, the
comparison on multiple threads shows how effective our approach is at achieving
concurrency. RDFox was specifically designed with that goal in mind in a
shared-memory setting. However, as one can see from our results, DMAT also
parallelises computation well: in some cases the speedup is larger than in the
case of RDFox. This seems to be the case mainly because data partitioning
allows each server to handle an isolated portion of the graph, which can reduce
the need for synchronisation.

\begin{table}[tb]
\centering
\caption{Comparison of Centralised and Distributed Reasoning}\label{tab:concurrency-v-distribution}
\vspace{-2ex}
\begin{tabular}{c|r|r|r|r|r|r|r|r}
\hline
    \multirow{3}{*}{}   & \multicolumn{8}{c}{Threads/Servers} \\
    \cline{2-9}
                        & \multicolumn{2}{c|}{1}    & \multicolumn{2}{c|}{2}    & \multicolumn{2}{c|}{4}    & \multicolumn{2}{c}{8} \\
    \cline{2-9}
                        & RDFox & DMAT              & RDFox & DMAT              & RDFox & DMAT              & RDFox & DMAT \\
    \hline
    Times (s)           & 86    & 256               & 56    & 140               & 35    & 82                & 16    & 53   \\
    Speed-up            & 1.0x  & 1.0x              & 1.5x  & 1.8x              & 2.5x  & 3.1x              & 5.4x  & 4.8x \\
    \hline
    Size                & \multicolumn{8}{c}{$134 M \rightarrow 182 M$} \\
    \hline
\end{tabular}
\vspace{4ex}
\caption{Comparison with WebPIE}\label{tab:WebPIE}
\vspace{-2ex}
\begin{tabular}{r|S|S|S[table-number-alignment = right]|S|S[table-number-alignment = right]|S}
    \hline
    \multirow{2}{*}{Dataset}    & \multicolumn{2}{c|}{Sizes (G)}                                & \multicolumn{2}{c|}{WebPIE (64 workers)}                      & \multicolumn{2}{c}{DMAT (12 servers)} \\
    \cline{2-7}
                                & \multicolumn{1}{c|}{Input}    & \multicolumn{1}{c|}{Output}   & \multicolumn{1}{c|}{Time (s)} & \multicolumn{1}{c|}{ktps/w}   & \multicolumn{1}{c|}{Time (s)} & \multicolumn{1}{c}{ktps/w} \\
    \hline
    4K                          & 0.5                           & 0.729                         & 1920                          & 4.1                           & 224                           & 85 \\
    8K                          & 1                             & 1.457                         & 2100                          & 7.5                           & 461                           & 81 \\
    36K                         & 5                             & 6.516                         & 3120                          & 24.9                          & 2087                          & 71 \\
    \hline
\end{tabular}
\vspace{4ex}
\caption{Scalability Experiments}\label{tab:scalability}
\vspace{-2ex}
\begin{tabular}{c|c|c|c|c|c}
    \hline
            &           & Input     & Output    & Time  & Rate      \\
    Workers & Dataset   & size (G)  & size (G)  & (s)   & (ktps/w)  \\
    \hline
    2       & 4K        & 0.5       & 0.73      & 646   & 212       \\
    6       & 12K       & 1.6       & 2.19      & 769   & 173       \\
    10      & 20K       & 2.65      & 3.64      & 887   & 151       \\
    \hline
\end{tabular}
\end{table}

\paragraph{Comparison with WebPIE.} Next, we compared DMAT with WebPIE to see
how our approach compares with the state of the art in distributed
materialisation. To keep the experimentation effort manageable, we did not
rerun WebPIE ourselves; rather, we considered the same input dataset sizes as
\citet{DBLP:journals/ws/UrbaniKMHB12} and reused their published results. The
setting of these experiments thus does not quite match our setting: (i)~WebPIE
handles only the ter Horst fragment of OWL and thus cannot handle all axioms in
the OWL 2 RL subset of the LUBM ontology; (ii)~experiments with WebPIE were run
on physical (rather than virtualised) servers with only 24~GB of RAM each; and
(iii)~WebPie used 64 workers, while DMAT used just 12 servers. Nevertheless, as
one can see from Table~\ref{tab:WebPIE}, despite using more than five times
fewer servers, DMAT is faster by an order of magnitude. Hadoop is a disk-based
system so lower performance is to be expected to some extent, but this may not
be the only reason: triples in DMAT are partitioned by subject so, unlike
WebPIE, DMAT does not perform any communication on subject--subject joins.

\paragraph{Scalability Experiments.} Finally, to investigate the scalability of
DMAT, we measured how the system's performance changes when the input data and
the number of servers increase proportionally. The results are shown in
Table~\ref{tab:scalability}. As one can see, increasing the size of the input
does introduce an overhead for each server. Our analysis suggests that this is
mainly because handling a larger dataset requires sending more messages, and
communication seems to be the main source of overhead in the system. This, in
turn, leads to a moderate reduction in throughout. Nevertheless, the system
still exhibits very high inferences rates and clearly scales to very large
inputs.

%% file: conclusion.tex
\section{Conclusion}\label{sec:conclusion}

In this paper we have presented a novel approach to datalog reasoning in
distributed RDF systems. Our work extends the distributed query answering
algorithm by \citet{pmnh18dynamic-rdf-stores}, from which it inherits several
benefits. First, the servers in our system are asynchronous, which is
beneficial for concurrency. Second, dynamic data exchange is effective at
reducing network communication, particularly when input data is partitioned so
that related triples are co-located. Finally, we have shown empirically that our prototype implementation
is an order of magnitude faster than WebPIE \cite{DBLP:conf/aaai/UrbaniJK16},
and that it scales to increasing data loads. In the near future we intend to conduct tests on a broader range of datasets and rule sets, as well as direct comparisons with other in-memory distributed systems that perform tasks similar to DMAT.

We see several interesting avenues for our future work. First, better
approaches to partitioning the input data are needed: hash partitioning does
not guarantee that joins other than subject--subject ones are processed on one
server, and graph partitioning cannot handle large input graphs. Second,
supporting more advanced features of datalog, such as stratified negation and
aggregation is also needed in many practical applications.

%% file: appendix.tex
\section{Proofs}

\lemmaccc*

\begin{proof}
Consider an arbitrary run of Algorithms~\ref{alg:materialisation-i}
and~\ref{alg:materialisation-ii}, arbitrary server $k$, and arbitrary facts
$f_1$ and $f_2$.

Assume that ${\mathsf{PAR}_k(f_1) \rightsquigarrow \mathsf{add}_k(f_2)}$ holds.
Then, after the call to $\textsc{Synchronise}$ in
line~\ref{process-message-partial-sync}, the local clock of server $k$ has a
value that is strictly larger than $T(f_1)$. Thus, when $f_2$ is assigned a
timestamp, ${T(f_2) > T(f_1})$ holds.

If ${\mathsf{process}_k(f_1) \rightsquigarrow \mathsf{FCT}_k(f_2)}$ (resp.\
${\mathsf{PAR}_k(f_1) \rightsquigarrow \mathsf{FCT}_k(f_2)}$) holds, then after
the call to $\textsc{Synchronise}$ in line \ref{process-fact-sync} (resp.\
\ref{process-message-partial-sync}), the local clock of server $k$ has a value
that is strictly larger than $T(f_1)$. Server $k$ reads this value into the
$\mathsf{FCT}$ message for $f_2$ in line~\ref{finish-match-send-FCT} or
line~\ref{process-message-fact-gettime}. Before $f_2$ is added on some
destination server, this server calls $\textsc{Synchronise}$ in
line~\ref{process-message-fact-sync}, which ensures $T(f_2) > T(f_1)$. \qed
\end{proof}

\lemmaocc*

\begin{proof}
Fix $r$ in the domain of $\mu_{k,\pi}$, we want to show that $ j\in
\mu_{k,\pi}(r) \rightsquigarrow r\in voc_{\pi}(I_j)$. If the initial state of
$voc_{\pi}(I_j)$ includes $r$, then there is nothing to prove because the
occurrence mappings are initialized consistent.

Now let $\mathsf{addOcc}_k(r,\pi, j)$ denote the point at line
\ref{process-message-fact-add-occurrence} where $j$ is added to the image of
$\mu_{k,\pi}(r)$. Let $f_{k,r}$ denote the first fact added to $I_k$ with $r$
as an argument and $\pi_k$ the position of $r$ in $f_{k,r}$. Let $f_{j,r}$ be
the analogous for $I_j$. Let $D(f)$ denote the set of update servers compiled
for fact $f$ in $\textsc{FinishMatch}$.

Assume $k\neq j$ and that $\mathsf{add}_k(f_{k,r})\rightsquigarrow
\mathsf{add}_j(f_{j,r})$ and $f_{j,r}\not\in I_0$. We know that the set
$D(f_{j,r})\cap D(f_{k,r})$ is not empty because of
Lemma~\ref{lemma:updateset}. Let $l$ be an element of the intersection. Both
$\mathsf{FCT}_l(f_{k,r})$ and $\mathsf{FCT}_l(f_{j,r})$ will occur on server
$l$. If $\mathsf{addOcc}_{l}(r,\pi_k,k)\rightsquigarrow
\mathsf{addOcc}_{l}(r,\pi_j,j)$, $k$ is be added to $D(f_{j,r})$ at line
\ref{process-message-fact-add-occurrence}, if it was not an initial member of
$D(f_{j,r})$. In either case, or in the case that $l=k$, $f_{j,r}$ updates
server $k$ before being added to $I_k$, hence $\mathsf{addOcc}_k(r, \pi_j,
j)\rightsquigarrow \mathsf{add}_j(f_{j,r})$. If instead
$\mathsf{addOcc}_{l}(r,\pi_j,j)\rightsquigarrow
\mathsf{addOcc}_{l}(r,\pi_k,k)$, $j$ is added to $D(f_{k,r})$, and we can still
conclude $\mathsf{addOcc}_k(r, \pi_j, j)\rightsquigarrow
\mathsf{add}_k(f_{k,r}) \rightsquigarrow \mathsf{add}_j(f_{j,r})$ because of
our initial assumption. If we assume
$\mathsf{add}_j(f_{j,r})\rightsquigarrow\mathsf{add}_k(f_{k,r})$ then the
symmetrical argument applies.

We have shown that, for $k\neq j$, $\mathsf{addOcc}_k(r, \pi_j, j)
\rightsquigarrow \mathsf{add}_j(f_{j,r})$, which is equivalent to the thesis.
When $k=j$, then it is true by construction that $\mathsf{addOcc}_k(r, \pi_k,
k) \rightsquigarrow \mathsf{add}_k(f_{k,r})$ and we conclude. \qed
\end{proof}

\begin{restatable}{lemma}{lemmaupdateset}\label{lemma:updateset}
Let $a,b\in P^{\infty}(I)$ have a common resource $r$, then $D_r(a)\cap
D_r(b)\neq \emptyset$, where $D_r(x)$ is the update set created for the
$\mathsf{FCT}$ message of $x$.
\end{restatable}

\begin{proof}
The proof is by induction over $P^i(I)$. It is true by definition for $i=0$
because the condition in Lemma~\ref{lemma:occ} has to apply to the initial
configuration. For the inductive step, suppose the property is true for all
facts in $P^{i-1}(I)$ and let $a,b\in P^i(I)$ share a resource $r$. We can find
two chains $a_0, a_1, \cdots, a_n = a$ and $b_0, b_1, \cdots, b_m = b$ such
that for each $i$ and $x\in\{a,b\}$ $x_i$ participates in the derivation of
$x_{i+1}$, it has $r$ among its arguments, and $D_r(x_i)\subseteq D_r(x_{i+1})$
(we select the branch of the derivation tree that corresponds to the first
matchings of the resource $r$ so that the partial mappings are passed between
$x_i$ and $x_{i+1}$). At each stage $D_r(a_i)\cap D_r(b_i)\subset
D_r(a_{i+1})\cap D_r(b_{i+1})$, therefore the property holds.
\end{proof}

\thmcorrectness*

\begin{proof}[Soundness]
The proof is by induction on the construction of sets $I_i$. The argument is
straightforward so we just present a sketch: when ${(\sigma, a, Q, h)}$ is
returned on some server $k$ in line~\ref{process-fact-match-rules},
substitution $\sigma$ satisfies ${a\sigma \in I_k}$; moreover, as matching of
$Q$ progresses, each substitution $\sigma'$ returned in line
line~\ref{process-message-eval} satisfies ${a_i\sigma' \in I_{k'}}$;
consequently, each substitution $\sigma$ in line~\ref{finish-match-send-FCT} is
an answer to the annotated query $Q$. Thus, each such $\sigma$ matches all body
atoms of the rule corresponding to ${(\sigma, a, Q, h)}$ in $P^\infty(I)$, and
so we clearly have ${h\sigma \in P^\infty(I)}$. \qed
\end{proof}

\begin{proof}[Completeness]
Let $P$ be a program, let $I$ be an input dataset, and let ${I_1, \dots,
I_\ell}$ be the datasets computes after Algorithms~\ref{alg:materialisation-i}
and~\ref{alg:materialisation-ii} finish on some partition of $I$ to $\ell$
servers. Our claim follows from the following property:
\begin{quote}
     ($\ast$) for each $i$ and each fact ${f \in P^i(I)}$, a server $k$ exists
     were ${f \in I_k}$ holds.
\end{quote}
The proof is by induction on $i$. The base case holds trivially, so we assume
that ($\ast$) holds for some ${i \geq 0}$ and show that it also holds for
${i+1}$. To this end, we consider an arbitrary fact ${f \in P^{i+1}(I)
\setminus P^i(I)}$. This fact is derived by a rule ${h \leftarrow b_0 \wedge
\dots \wedge b_n \in P}$ and substitution $\sigma$ such that ${h\sigma = f}$
and ${b_j\sigma \in P^i(I)}$ for ${0 \leq j \leq n}$. Now choose $p$ as the
smallest integer between $0$ and $n$ such that ${T(b_{p'}) \leq T(b_p\sigma)}$
holds for each ${0 \leq p' \leq n}$. Now let ${a_0, \dots, a_n}$ be the body
atoms of the rule rearranged so that ${a_0 = b_p}$ is the pivot atom, and the
remaining atoms correspond to the annotated query ${Q = a_1^{\bowtie_1} \wedge
\dots \wedge a_n^{\bowtie_n}}$ returned by ${\textsc{MatchRules}(b_p\sigma,
P)}$ in line~\ref{process-fact-match-rules} on fact $b_p\sigma$. Finally, for
each ${0 \leq j \leq n}$, let $\sigma_j$ be the substitution $\sigma$
restricted to all variables occurring in atoms ${a_0, \dots, a_j}$ and let
${\tau_j = T(a_j\sigma)}$; moreover, ($\ast$) holds for $i$ by the induction
assumption, so there exists a server $k_j$ such that ${a_j\sigma \in I_{k_j}}$
holds. We next prove the following:
\begin{quote}
    ($\lozenge$) for each ${0 \leq j \leq n}$, function
    ${\textsc{FinishMatch}(j, \sigma_j, a_j, Q, h, \tau_0, \pmb\lambda_j)}$ is
    called for some $\pmb\lambda_j$.
\end{quote}
Property ($\lozenge$) implies our claim because in
lines~\ref{finish-match-matched-start}--\ref{finish-match-matched-end} the
algorithm then constructs a $\mathsf{FCT}$ message for $h\sigma$ and dispatches
it to some server $k_h$, so $h\sigma$ is eventually added to $I_{k_h}$ in
line~\ref{process-message-add}, as required for ($\ast$).

We next prove ($\lozenge$) by induction on ${0 \leq j \leq n}$. For the base
case, ${a_0\sigma \in I_{k_0}}$ ensures that $\textsc{ProcessFact}(a_0\sigma)$
is called on server $k_0$, so ${\textsc{MatchRules}(a_0\sigma, P)}$ returns
${(\sigma_0, a_0, Q, h)}$, and ${\textsc{FinishMatch}(0, \sigma_0, a_0, Q, h,
\tau_0, \pmb\emptyset)}$ is called in line~\ref{process-fact-finish-match}. For
the induction step, we assume that ($\lozenge$) holds for some ${0 \leq j <
n}$, and we show that it holds for ${j+1}$ as well. To this end, we consider
several cases.

Assume that event ${\mathsf{PAR}_{k_{j+1}}(a_0\sigma, j+1)}$ occurs at some
point during the algorithm's run. Server $k_{j+1}$ then executes
line~\ref{process-message-eval} for $a_{j+1}^{\bowtie_{j+1}}$. Note that
${a_{j+1}\sigma \in I_{k_{j+1}}}$ holds by induction assumption. We next show
that server $k_{j+1}$ contains $a_{j+1}\sigma$ at the point in time when
line~\ref{process-message-eval} is executed. We have the following
possibilities.
\begin{itemize}
    \item If event $\textsf{add}_{k_{j_1}}(a_{j+1}\sigma)$ never happens, then
    server $k_{j+1}$ contains fact $a_{j+1}\sigma$ since the algorithm's start.
    
    \item If ${\textsf{add}_{k_{j_1}}(a_{j+1}\sigma) \rightsquigarrow
    \mathsf{PAR}_{k_{j+1}}(a_0\sigma, j+1)}$ holds, then server $k_{j+1}$
    clearly contains fact $a_{j+1}\sigma$ at this point in time.
    
    \item If ${\mathsf{PAR}_{k_{j+1}}(a_0\sigma, j+1) \rightsquigarrow
    \textsf{add}_{k_{j_1}}(a_{j+1}\sigma)}$ were to hold, then
    Lemma~\ref{lemma:ccc} implies ${T(a_0\sigma) < T(a_{j+1}\sigma)}$,
    contradicting our assumption that ${T(a_{j+1}\sigma) \leq T(a_0\sigma)}$.
\end{itemize}
Moreover, if ${T(a_{j+1}\sigma) = T(a_0\sigma)}$, since ${a_0 = b_p}$ was
chosen so that $p$ is the least index of a body atom matched to a fact with
timestamp $T(a_0\sigma)$, the shape of $Q$ from
\eqref{eq:annotated-query-for-p} ensures that ${\bowtie_{j+1} = \; \leq}$.
Consequently, the call to $\textsc{Evaluate}$ in
line~\ref{process-message-eval} on server $k_{j+1}$ returns $\sigma_{j+1}$, so
the call in line~\ref{process-message-finish-match} ensures ($\lozenge$).

Now assume that event ${\mathsf{PAR}_{k_{j+1}}(a_0\sigma, j+1)}$ never occurs
during the algorithm's run---that is, server $k_j$ never forwards a
$\mathsf{PAR}$ message to server $k_{j+1}$. Then, for some ${\pi \in \Pi}$ and
${r = \term{a_{j+1}\sigma_j}{\pi}}$, we have ${k_{j+1} \not\in \lambda_\pi(r)}$
at the point in time when line~\ref{finish-match-destination-list} is executed
on server $k_j$, ensuring that $k_{j+1}$ is removed from $D$. However, this
$\lambda_\pi(r)$ is populated in line~\ref{finish-match-add-supporting-ocr}
when resource $r$ is matched on some server $k_s$ with ${0 \leq s \leq j}$, so
at that point in time we have ${k_{j+1} \not\in \mu_{k_s,\pi}(r)}$. Now if
event $\textsf{add}_{k_{j+1}}(a_{j+1}\sigma)$ never happened, then server
$k_{j+1}$ would contain $a_{j+1}\sigma$ when the algorithm starts; but then,
since $\mu_{k_s,\pi}$ is defined on $r$, Lemma~\ref{lemma:occ} implies
${k_{j+1} \not\in \mu_{k_s,\pi}(r)}$, which is a contradiction. Consequently,
event $\textsf{add}_{k_{j+1}}(a_{j+1}\sigma)$ occurs on server $k_{j+1}$.

Moreover, let ${\alpha = \mathsf{process}_{k_s}(a_0\sigma)}$ if ${s = 0}$, and
let ${\alpha = \mathsf{PAR}_{k_s}(a_0\sigma,j)}$ if ${s > 0}$. Function
$\textsc{FinishMatch}$ is called on server $k_s$ by the induction assumption
for ($\lozenge$), so event $\alpha$ occurs on server $k_s$.

Now note that the set $I_{k_{j+1}}\cap \{f\mid \term{f}{\pi}=r\}$ is not empty
and let $f_{r,\pi}$ be the first fact in the in the set that is added to
$k_{j+1}$ and consider the set the set $D$ constructed in
line~\ref{finish-match-D-n} for $f_{r,\pi}$ (from here on \textit{update set}).
By definition,the partial occurrences for resource $r$ at position $\pi$ in the
$\mathsf{FCT}$ message for $f_{r,\pi}$ cannot contain $k_{j+1}$, therefore the
occurrences of $r$ are added to $D$. Let $D_r$ be the contribution of $r$ to
$D$; it cannot be empty because $f_{r,\pi}$ was derived at some server and it
is either the case that $D_r$ contains $k_s$ or not.

First, assume $k_s\in D_r$, then event $\mathsf{FCT}_{k_s}(f_{r,\pi})$ occurs
on server $k_s$. This event updates occurrences for $r$ on $k_s$, so for
${k_{j+1} \not\in\mu_{k_s,\pi}(r)}$ to hold, ${\alpha \rightsquigarrow
\mathsf{FCT}_{k_s}(a_{j+1}\sigma)}$ must hold. But then, regardless of how
$\alpha$ is defined, Lemma~\ref{lemma:ccc} implies $T(a_0\sigma) <
T(a_{j+1}\sigma)$, which contradicts our assumption that ${T(a_{j+1}\sigma)
\leq T(a_0\sigma)}$ holds.

Otherwise, we know that the presence of the resource $r$ in $I_{k_s}$ had not
yet been fully notified when the the partial mappings for the argument $r$ of
$f_{r,\pi}$ were created. Let $f'_r$ be the first added fact to $I_{k_s}$ that
contains $r$ in any position, and its update set $D'$. We apply
Lemma~\ref{lemma:updateset} on $f_{r,\pi}$ and $f'_r$ to find an element in
$k'\in D_r\cap D'_r$ where both $\mathsf{FCT}_{k'}(f_{r,\pi})$ and
$\mathsf{FCT}_{k'}(f'_r)$ occur. If
$\mathsf{FCT}_{k'}(f_{r,\pi})\rightsquigarrow \mathsf{FCT}_{k'}(f'_r)$, then
$k_{j+1}$ is added to $\mu_{k',\pi}$ at line
\ref{process-message-fact-add-occurrence} when the first message is processed,
and then to $\lambda_{\pi}$ at line \ref{process-message-fact-add-occurrence}
when the second message is processed. This event is then followed by
$\mathsf{FCT}_{k_s}(f'_r)$ and we proceed as done with
$\mathsf{FCT}_{k_s}(f_{r,\pi})$ above to derive $T(a_0\sigma) < T(f'_r)$ from
Lemma~\ref{lemma:ccc}. But $f'_r$ is the first appearance of $r$ in server
$k_s$, hence $T(a_0\sigma) < T(a_s\sigma)$ which contradicts our assumption
that $T(a_s\sigma) < T(a_0\sigma)$ holds.

If $\mathsf{FCT}_{k'}(f'_r) \rightsquigarrow \mathsf{FCT}_{k'}(f_{r,\pi})$,
then $k_s$ is added to $\mu_{k',\pi}$ at line
\ref{process-message-fact-add-occurrence} when the first message is processed,
and then to $D$ at line \ref{process-message-fact-extend-D} when the second
message is processed, so $\mathsf{FCT}_{k_s}(f_{r,\pi})$ eventually happens and
we conclude like above.

In summary, we proved that ${\mathsf{PAR}_{k_{j+1}}(a_0\sigma, j+1)}$ occurs on
server $k_{j+1}$ and this concludes the theorem. \qed
\end{proof}

\begin{proof}[Nonrepetition of Derivations]
Assume that $\textsc{processFact}$ considers two facts $f_1$ and $f_2$, both of
which matched the same rule and produce the same substitution $\sigma$. Let
$b_1$ and $Q_1$ be the pivot atom and the annotated query returned in
line~\ref{process-fact-match-rules} when $f_1$ is processed, and let $b_2$ and
$Q_2$ be defined analogously. Thus, ${b_1\sigma = f_1}$ and ${b_2\sigma =
f_2}$. Since each fact is processed only once, atoms $b_1$ and $b_2$ are
distinct. Now w.l.o.g.\ let us assume that $b_1$ occurs before $b_2$ in the
body of the rule; thus, the atom corresponding to $b_2$ in $Q_1$ is annotated
with $\leq$, and the atom corresponding to $b_1$ in $Q_2$ is annotated with
$<$. But then, $f_2$ is not matched by $Q_1$ if ${T(f_1) < T(f_2)}$ holds, and
$f_1$ is not matched by $Q_2$ if ${T(f_1) \geq T(f_2)}$ holds, which
contradicts our assumption that the algorithm repeats inferences. \qed
\end{proof}